\documentclass[journal]{IEEEtran}

\ifCLASSINFOpdf
\else
   \usepackage[dvips]{graphicx}
\fi
\usepackage{url}
\usepackage[caption=false,font=footnotesize]{subfig}

\usepackage{pgfplots}
\usepackage{graphicx}
\usepackage{amsmath}
\usepackage{xcolor,color}
\usepackage{amsfonts}
\usepackage{bm}
\usepackage{amsmath,amsthm,amssymb,rotating, mathrsfs}
\usepackage{mathtools}
\usepackage{cancel} 
\usepackage{cite}

\usepackage{xfrac} 
\usepackage{xparse} 

\usepackage{pgfplots}
\usepackage{graphicx}
\usepackage{amsmath}
\usepackage{xcolor,color}
\usepackage{amsfonts}
\usepackage{bm}
\usepackage{amsmath,amsthm,amssymb,rotating, mathrsfs}
\usepackage{mathtools}
\usepackage{cancel} 
\usepackage{cite}

\usepackage{xfrac} 
\usepackage{xparse} 

\usepackage{booktabs}

\newcommand{\cF}{\mathcal{F}}

\newcommand{\cI}{\mathcal{I}}

\newcommand{\cN}{\mathcal{N}}
\newcommand{\cO}{\mathcal{O}}
\newcommand{\cP}{\mathcal{P}}

\newcommand{\cR}{\mathcal{R}}

\newcommand{\cT}{\mathcal{T}}

\newcommand{\bbR}{\mathbb{R}}

\newcommand{\argmin}{\mathop{\rm argmin}}

\NewDocumentCommand{\norm}{mG{2}}{\lVert#1\rVert_{#2}}

\newcommand{\trans}{\top}
\newcommand{\trsp}[1]{#1^\trans}

\DeclareMathOperator{\conv}{conv}

\newcommand{\myparagraph}[1]{\smallskip\noindent\textbf{#1.}}

\newtheorem{lemma}{Lemma}
\newtheorem{prop}{Proposition}

\theoremstyle{definition}

\usepackage{hyperref} 
\hypersetup{hidelinks}

\begin{document}

\title{Linear Regression without Correspondences via Concave Minimization}

\author{Liangzu Peng \ \ \ \ \ \ \ \ \ \ \ \ \  Manolis C. Tsakiris
\thanks{The authors are with the School of Information Science and Technology, ShanghaiTech University, Shanghai, China (e-mail: \{penglz,mtsakiris\}@shanghaitech.edu.cn).}
\thanks{This manuscript is the preprint version of the letter available at \url{https://ieeexplore.ieee.org/document/9178410}.}
}

\markboth{PREPRINT}
{Shell \MakeLowercase{\textit{et al.}}: Bare Demo of IEEEtran.cls for IEEE Journals}
\maketitle

\begin{abstract}
Linear regression without correspondences concerns the recovery of a signal in the linear regression setting, where the correspondences between the observations and the linear functionals are unknown. The associated maximum likelihood function is NP-hard to compute when the signal has dimension larger than one. To optimize this objective function we reformulate it as a concave minimization problem, which we solve via branch-and-bound. This is supported by a computable search space to branch, an effective lower bounding scheme via convex envelope minimization and a refined upper bound, all naturally arising from the concave minimization reformulation. The resulting algorithm outperforms state-of-the-art methods for fully shuffled data and remains tractable for up to $8$-dimensional signals, an untouched regime in prior work.
\end{abstract}


\IEEEpeerreviewmaketitle

\section{Introduction}\label{sec:introduction}
Linear regression without correspondences is concerned with the estimation of an $n$-dimensional signal $x^* \in \bbR^n$ from a set of $m$ noisy linear measurements $y_i \in \bbR$ and the set of linear functionals $\trsp{a_j} \in \bbR^{1 \times n}$ that generated them, in the absence of the correspondence between these two. Concretely, for every $i \in [m]:=\{1,\dots, m\}$ we have $y_i =  \trsp{a}_{\pi^*(i)} x^* + \epsilon_i$, where $\pi^*$ is an unknown permutation of $[m]$ and $\epsilon_i$ is additive noise. With $y, \epsilon \in \bbR^m$ and $A \in \bbR^{m \times n}$ having $y_i, \epsilon_i$ and $\trsp{a}_{i}$ in their $i$-th row respectively and $\Pi^*$ a permutation matrix

\vspace{-0.1in}

\begin{align}\label{eq:SLR_model}
y=\Pi^* A x^* + \epsilon,
\end{align}
and the objective is to estimate $x^*$ from $y,A$.

One of the first theoretical papers addressing this problem in generality showed that in the noiseless case and under general position hypothesis on the entries of the matrix $A$, the problem is well-posed for any $x^*$ and has a unique solution as long as $m \ge 2n$ \cite{Unnikrishnan:2015,Unnikrishnan-TIT18}. If instead the signal is allowed
to be generic with respect to the measurements $A$ it was further shown that $m \ge n+1$ measurements are sufficient \cite{Tsakiris-arXiv18v2}. 
These results were then generalized to arbitrary linear transformations beyond permutations and down-samplings by \cite{Tsakiris-ICML2019,Tsakiris-arXiv18b,Tsakiris-arXiv18b-v5,Peng-arXiv2020}; see also \cite{Dokmanic-SPL2019}. Bringing back the noise $\epsilon$ into the picture \cite{Pananjady-TIT18} obtained SNR conditions under which recovery of $\Pi^*$ is possible from the maximal likelihood estimator

\begin{align}\label{eq:MLE}
(\hat{\Pi},\hat{x})\in \argmin_{x\in\bbR^n,\Pi\in\cP} \norm{\Pi y- A x},
\end{align}
where $\cP$ consists of all $m\times m$ permutation matrices. Finally, a convex $\ell_1$ recovery theory has been developed for the case where only a fraction of the correspondences is missing \cite{Slawski-JoS19}. 

Standing on firm theoretical grounds, in this paper we take an interest in the computational challenges of solving the linear regression without correspondences problem. The easiest case is when $n=1$ for which an $\cO(m\log m)$ sorting-based algorithm optimally solves \eqref{eq:MLE} \cite{Pananjady-TIT18}. The next tractable instance is when the data are partially shuffled, as occurs, e.g., in record linkage \cite{Slawski-JoS19,Lahiri-JASA2005,Shi-arXiv18,Slawski-UAI2019,Slawski-arXiv2019,Slawski-arXiv2019b,Zhang-arXiv2019}. In such a case, the estimation of $x^*$ may be performed via convex $\ell_1$ robust regression \cite{Slawski-JoS19} or a pseudo-likelihood approach \cite{Slawski-arXiv2019b}, these tolerate at most $50\%$ or $70\%$ mismatches, respectively. When $n \ge 2$ and the data are fully shuffled, as in point set registration \cite{Li-ICCV07,Lian-PAMI17} or signal estimation using distributed sensors \cite{Song-ISIT18,Peng-ICASSP2019,Zhu-CL2017}, \eqref{eq:MLE} is strongly NP-hard \cite{Pananjady-TIT18,Hsu-NIPS17}. Exhaustive search comes with $\cO(m!mn^2)$ complexity. Alternating minimization, iteratively updating $\Pi$ and $x$, is sensitive to initialization \cite{Abid:arXiv17, Abid:arXiv18, Haghighatshoar-TSP18}. The RANSAC-type algorithm of \cite{Elhami-ICASSP17} originally applicable only to noiseless data and later robustified by \cite{Tsakiris-ICML2019} to which we refer as \cite{Tsakiris-ICML2019}-B, requires solving $\cO\big(\binom{m}{n}n!\big)$ $n \times n$ linear systems of equations. The fully polynomial-time approximation scheme of \cite{Hsu-NIPS17} employs sophisticated enumeration using $\cO(m^{4n})$ $\epsilon$-nets thus entailing a complexity of at least $\cO(m^{4n})$. The algebraic-geometric algorithm of \cite{Tsakiris-arXiv18v2} uses Gr{\"o}bner basis methods to solve an $n \times n$ polynomial system of equations derived from \eqref{eq:SLR_model}. Even though it has linear complexity in $m$ its running time grows exponentially in the signal dimension $n$: for $m=500$ this is $15$ msec for $n= 4$, $45$ sec for $n=5$, $37$ min for $n=6$, and no result reported for $n\geq 7$.
Finally, a working algorithm \cite{Tsakiris-ICML2019}-A for unlabeled sensing was built in \cite{Tsakiris-ICML2019} to globally optimize \eqref{eq:MLE} by combining branch-and-bound and dynamic programming to repeatedly solve a \textit{one-dimensional} linear assignment problem in $\cO(m^2)$ time as opposed to the typical $\cO(m^3)$ of \cite{Jonker-Computing1987}. Even though with promising empirical performance and a variation that gave state-of-the-art results in image registration, the algorithm does not scale well for $n >4$ due to its \textit{naive lower bounding scheme} \cite{Tsakiris-ICML2019}.

Even though the known theoretical SNR requirements for correct recovery via \eqref{eq:MLE} are rather demanding \cite{Pananjady-TIT18}, optimizing \eqref{eq:MLE} can still be effective in reasonable real-data situations (Section \ref{section:numerics}). Hence in this paper we propose a branch-and-bound technique for solving \eqref{eq:MLE}. The main innovation here is the reformulation of \eqref{eq:MLE} into a concave minimization problem. This leads to a computable search space to branch, a tight lower bound via \textit{convex envelope} computations, and a refined upper bound through alternating minimization. To the best of our knowledge, the proposed algorithm is the best performing working method for fully shuffled data and remains tractable for $n=7,8$ and $m=100$, an untouched regime of prior work.




\section{The Concave Minimization Approach}\label{section:cvx_minimize}
We propose a concave minimization approach of the branch-and-bound type  to solve \eqref{eq:MLE}. The branch-and-bound algorithm is used to minimize a given objective function, say $g$, globally optimally \cite{Emiya-ICASSP14,Tsakiris-ICML2019,Falk-MS1969,Kalantari-MOR1987}. That is, the computed solution $\hat{z}$ is  $\delta$-close to the optimal $z^*$, i.e. $g(\hat{z})<g(z^*)+\delta$ for some $\delta>0$. Simply put, given an initial region containing $z^*$, this algorithm \textit{branches}: it recursively subdivides a selected region into sub-regions. On the other hand, \textit{bounding} is to determine the lower bound of $g$ over a given sub-region. In parallel, the algorithm computes an upper bound of $g(z^*)$ and accordingly the smallest upper bound $q_u$ among upper bounds obtained so far. A sub-region is excluded if its lower bound is not less than $q_u-\delta$.
In this way the algorithm explores and narrows the search space until a $\delta$-close solution is found. 
The tighter the lower and upper bounds, the more regions can be excluded and the faster the algorithm converges. 

Two challenges are in the way of adopting branch-and-bound for problem \eqref{eq:MLE}. First is the choice of the branching variable, $\Pi\in\cP$ or $x\in\bbR^n$. 
Both strategies have been explored in the literature. Branching over $\cP$ is far from feasible as discussed in \cite{Li-ICCV07}, even if more than $99.9\%$ permutations can be excluded, which is possible with a tight lower bound \cite{Emiya-ICASSP14} or with a learning-based pruning strategy \cite{Shen-TWC2019}. This is because $|\cP|=m!$ grows exponentially with $m$, for example, $10!>2^{21}, 20!>2^{61}$. On the other hand,  \cite{Tsakiris-ICML2019}-A proposes to branch over $\bbR^n$, but it requires as a hyper-parameter a region that contains the global minimizer. The second challenge involves the trade-off between the efficient computation and tightness of the lower and upper bounds. For example, \cite{Tsakiris-ICML2019}-A uses dynamic programming to efficiently compute a rather loose bound in $\cO(m^2)$ time, while \cite{Emiya-ICASSP14} computes a tight bound by solving an expensive convex optimization problem.

In this work we reformulate \eqref{eq:MLE} into the minimization of a quadratic concave function $g$ over a convex polytope $\cF^\circ\subset \bbR^n$. This type of problem is a classic one already studied in \cite{Falk-MS1969}, where branching over $\cF^\circ$ was proposed. It was observed later in \cite{Kalantari-MOR1987} that it is more efficient to branch over the smallest rectangle $\cR^\circ$ that contains $\cF^\circ$ than directly over $\cF^\circ$. In our case $\cR^\circ$ can be computed via solving $2n$ sorting problems. Our branching space is this easily computable rectangle.

It is also this reformulation that leads to a  balance between efficiency and tightness of the lower bounding strategy. Following \cite{Kalantari-MOR1987}, we obtain tight lower bounds by minimizing the convex envelope of $g$ over sub-rectangles of $\cR^\circ$, which amounts to solving linear assignment problems. We note here that the classic idea of \cite{Kalantari-MOR1987} has recently been applied with good performance to image registration \cite{Lian-PAMI17} and multi-target tracking \cite{Ji-IEEE2019}. Compared to \cite{Lian-PAMI17} and \cite{Ji-IEEE2019}, our reformulation avoids directly manipulating a large $n\times m^2$ matrix, while \cite{Lian-PAMI17} and \cite{Ji-IEEE2019} perform QR decomposition of a matrix of such or larger size. Finally, upper bounds are obtained by a suitably initialized alternating minimization scheme. This is a further improvement upon the typical upper bound computation of \cite{Lian-PAMI17,Ji-IEEE2019}. As we will see in Section \ref{section:numerics} (Table \ref{table:fully_shuffled_rt}), this leads to an algorithm that outperforms existing algorithms for linear regression with fully shuffled data.




\subsection{Concave Minimization Reformulation}\label{subsection:obj}
Let $A=U_A\Sigma_A \trsp{V}_A$ be the thin SVD of the rank-$r$ matrix $A\in\bbR^{m\times n}$.
For solving \eqref{eq:MLE} we consider the following problem
\begin{align}
(\hat{\Pi},\hat{w})\in \argmin_{\Pi\in \cP}\min_{w\in\bbR^n}\norm{\Pi y- U_Aw}, \label{eq:MLE2}
\end{align}
With $(\hat{\Pi},\hat{w})$ of \eqref{eq:MLE2} we can obtain $\hat{x}$ of \eqref{eq:MLE} by solving $\hat{w}=\Sigma_A \trsp{V}_A x$ for $x$, a linear system of equations that have exactly one solution if $r=n$ and have infinitely many if $r<n$. The solution to the inner minimization of \eqref{eq:MLE2} is $w_{\Pi} =\trsp{U_A} \Pi y$. With $\bar{y}:=y/\norm{y}$ and the Kronecker product $\otimes$, \eqref{eq:MLE2} is the same as
\begin{align}
&\hat{\Pi}\in \argmin_{\Pi\in \cP}\norm{\Pi y- U_A\trsp{U_A}\Pi y} \label{eq:cvx_disguise}\\
\Leftrightarrow&\hat{\Pi}\in\argmin_{\Pi\in \cP} \norm{y}^2-\trsp{y}\trsp{\Pi} U_A\trsp{U}_A\Pi y\\
\Leftrightarrow&\hat{\Pi}\in\argmin_{\Pi\in \cP} -\trsp{\bar{y}}\trsp{\Pi} U_A\trsp{U}_A\Pi\bar{y} \\
\Leftrightarrow&\hat{\Pi}\in\argmin_{\Pi\in \cP} -\norm{(\trsp{\bar{y}}\otimes\trsp{U}_A)\text{vec}(\Pi)}^2. 
\label{eq:opt_over_Pi}
\end{align}
As already mentioned, branching over $\cP$ to solve \eqref{eq:opt_over_Pi} is not feasible. One may instead optimize \eqref{eq:opt_over_Pi} over 
\begin{align}\label{eq:ds}
\conv(\cP)=\{B\in\bbR^{m\times m} : \trsp{B}e=e, Be=e,B\geq 0 \},
\end{align}
the convex hull of $\cP$. In \eqref{eq:ds}, $B\geq 0$ denotes that all entries of $B$ are no less than $0$, and $e$ is the $m$-dimensional vector whose entries are $1$. Note that $\conv(\cP)$ is the well-known Birkhoff polytope, consisting of the set of all $m\times m$ doubly stochastic matrices \cite{Birkhoff-1946,Neumann-1953}. So we arrive at
\begin{align}\label{eq:cvx_max}
\min_{B\in \conv(\cP)} -\norm{K\text{vec}(B)}^2=:f(B),
\end{align}
where  $K=\trsp{\bar{y}}\otimes\trsp{U}_A$. 
The relationship between the minimizers of \eqref{eq:opt_over_Pi} and \eqref{eq:cvx_max} is characterized by the following proposition. 
\begin{prop}\label{prop:exist_permutation}
	If \eqref{eq:opt_over_Pi} has a unique minimizer $\hat{\Pi}$, then $\hat{\Pi}$ is also the unique minimizer for \eqref{eq:cvx_max}. 
\end{prop}
\begin{proof}
	Any minimizer $\hat{B}\in\conv(\cP)$ of \eqref{eq:cvx_max} can be written as a convex combination of permutation matrices, say $\hat{B}=\sum_{i=1}^d \lambda_i\Pi_i$, $\lambda_i\geq 0$ and $\sum_{i=1}^{d}\lambda_i=1$ \cite{Birkhoff-1946,Neumann-1953}. Since $f$ of \eqref{eq:cvx_max} is concave we have $\sum_{i=1}^{d}\lambda_if(\Pi_i)\leq f(\hat{B})$. Suppose that there is some $\Pi_j\neq \hat{\Pi}$. Since \eqref{eq:opt_over_Pi} and \eqref{eq:cvx_max} have the same objective $f$, we get $f(\hat{\Pi})\leq f(\Pi_i)$ for $i\in[d]$ and $f(\hat{\Pi})<f(\Pi_j)$.
	So $f(\hat{\Pi})=\sum_{i=1}^{d}\lambda_if(\hat{\Pi})< \sum_{i=1}^{d}\lambda_i f(\Pi_i)$, which implies $f(\hat{\Pi})<f(\hat{B})$, a contradiction. Hence $\hat{B}=\hat{\Pi}$.
\end{proof}
In what follows we assume that the minimizer of \eqref{eq:opt_over_Pi} is unique. Proposition \ref{prop:exist_permutation} allows us to compute the desired solution to \eqref{eq:opt_over_Pi} by instead optimally solving \eqref{eq:cvx_max}.
Note that $\conv(\cP)$ is a polytope of dimension\footnote{The dimension of a polytope in $\bbR^d$ is the dimension of the smallest affine subspace of $\bbR^d$ containing that polytope \cite{Burkard-AP09}. Note that there are $(2m-1)$ linearly independent equations in \eqref{eq:ds}.}  $(m-1)^2$. High dimensionality suggests inefficiency of branching over $\conv(\cP)$. Next we show that the branching can be conducted over a convex polytope of dimension $n\ll (m-1)^2$. 

The $n\times m^2$ matrix $K$ in \eqref{eq:cvx_max} is of rank $n$. Write $\{\sigma_i\}_{i=1}^n$ and $\{v_i\}_{i=1}^n$ for its singular values and right singular vectors. We can decompose $f(B)$ into a sum of $n$ quadratic terms: 
\begin{align}
f(B) = -\sum\nolimits_{i=1}^n (\sigma_i\trsp{v_i}\text{vec}(B))^2.
\end{align}
Hence minimizing \eqref{eq:cvx_max} is equivalent to
\begin{align}\label{eq:cvx_max_sum}
\min_{\substack{z_i=\sigma_i\trsp{v_i}\text{vec}(B ),\ B \in \conv(\cP)}} -\sum\nolimits_{i=1}^n z_i^2=:g(z_1,\dots,z_n).
\end{align}
In \eqref{eq:cvx_max_sum} $g$ is  concave in $n$ variables $[z_1,\dots,z_n]=:z\in\bbR^n$. Although Problems \eqref{eq:cvx_max} and \eqref{eq:cvx_max_sum} are equivalent, the objective function $g$ is surprisingly simpler. This will play a key role in the sequel. We proceed with three key remarks.

First, arriving at \eqref{eq:cvx_max_sum} is cheap. It requires
computing the singular values $\{\sigma_i\}_{i=1}^n$ and vectors $[v_1,\dots,v_n]=:V$ of the large matrix $K$. This otherwise inefficient SVD computation is reduced to a simple Kronecker product, owing to:
\begin{lemma}\label{lemma:kron_SVD}
	$V=\bar{y}\otimes U_A$, and $\sigma_i=1,i\in \{1,\dots, n\} =:[n]$.
\begin{proof}
		With the $n\times n$ identity matrix $I_n$, the thin SVDs of $\trsp{\bar{y}}$ and $\trsp{U}_A$ are $1\cdot 1 \cdot \trsp{\bar{y}}$ and $I_n\cdot I_n\cdot \trsp{U}_A$ respectively. Hence the thin SVD of $K=\trsp{\bar{y}}\otimes\trsp{U}_A$ is $(1\otimes I_n)(1\otimes I_n)\trsp{(y\otimes U_A)}$. 
\end{proof} 
\end{lemma}

Secondly, solving \eqref{eq:cvx_max_sum} is cheap for $n=1$. In this case the objective function is $-(\trsp{v_1}\text{vec}(B))^2$, so it suffices to solve
\begin{align}\label{eq:bounds-rectangle}
\min_{B \in \conv(\cP)} \trsp{\bar{y}}Bu_{A,1} \text{\ \ \ and\ \ \ } \max_{B \in \conv(\cP)} \trsp{\bar{y}}Bu_{A,1},
\end{align}
where $u_{A,1}$ is the first column of $U_A$ and $v_1=\bar{y}\otimes u_{A,1}$. To maximize $\trsp{\bar{y}}Bu_{A,1}$ over $B \in \conv(\cP)$ we can instead maximize it over $\cP$ since the former has some permutaion as its optimal solution. The latter is equivalent to maximizing $\trsp{\bar{y}}_{\uparrow}Bu_{A,1}$ over $\cP$, where $\trsp{\bar{y}}_{\uparrow}$ consists of the entries of $y$ arranged in ascending order. Letting $\trsp{\bar{y}}_{\downarrow}$ record the entries of $y$ in descending order we see through a similar lens that the left problem of \eqref{eq:bounds-rectangle} can be solved by minimizing $\trsp{\bar{y}}_{\downarrow}Bu_{A,1}$ over $\cP$. What comes into play next is the \textit{rearrangement inequality}, which states that it is the permutation bringing $u_{A,1}$ to $(u_{A,1})_{\uparrow}$ that maximizes $\trsp{\bar{y}}_{\uparrow}Bu_{A,1}$ and minimizes $\trsp{\bar{y}}_{\downarrow}Bu_{A,1}$ over $\cP$ simultaneously. To conclude we can solve \eqref{eq:bounds-rectangle} via sorting.

Finally, it is also cheap to compute the smallest rectangle $\cR^\circ$ that contains the constraint set of \eqref{eq:cvx_max_sum}, the latter being
\begin{align}
\cF^\circ=\{z\in\bbR^n: z=\trsp{V}\text{vec}(B),\ B\in\conv(\cP) \}.
\end{align}
As already noted, it is over $\cR^\circ$ that we branch. We compute $\cR^\circ$ as follows.
For $i\in[n]$, we have $z_i=\trsp{v_i}\text{vec}(B)$ and
	\begin{align}\label{eq:lb_i}
	\min_{B \in \conv(\cP)} \trsp{\bar{y}}Bu_{A,i} \leq z_i\leq  \max_{B \in \conv(\cP)} \trsp{\bar{y}}Bu_{A,i},
	\end{align}
	where we note that $v_i=\bar{y}\otimes u_{A,i}$. The minimum $l_i^\circ$ and maximum $u_i^\circ$ of $z_i$ can be computed by solving the two problems in \eqref{eq:lb_i} respectively via sorting. So $\cR^\circ$ is given by
\begin{align}\label{eq:initial_rectangle}
\cR^\circ = \{z\in\bbR^n : z_i\in [l_i^\circ, u_i^\circ],\ i\in[n] \}.
\end{align}

\subsection{The Lower Bounding Scheme}\label{subsection:lb}
Each iteration of the branching algorithm involves some sub-rectangle
$\cR= \{z\in\bbR^n: z_i\in \cI_i=[l_i, u_i],\ i\in[n] \}$ of $\cR^\circ$. We discuss how to determine a lower bound of $g=-\sum_{i=1}^n z_i^2$ over the feasible set $\cF^\circ \cap \cR$. 

Our lower bound computation is intimately related to the notion of \textit{convex envelop} of some function $h$ defined on a set $\cT$, denoted by $\conv_\cT(h)$, which is the largest convex function majorized by $h$ on $\cT$. Geometrically, $\conv_\cT(h)$ is the point-wise supremum of all affine functions bounded above by $h$ on $\cT$. This immediately gives us a formula for $\conv_\cR(g)$:
\begin{lemma}
	The convex envelope of $g$ over $\cR$ is given by
	\begin{align}
	\conv_\cR(g)(z)=\sum\nolimits_{i=1}^nl_iu_i-\sum\nolimits_{i=1}^n(l_i+u_i)z_i. \label{eq:ce}
	\end{align}
\end{lemma}
\begin{proof}
	Let $g_i(z_i)=-z_i^2$. Then $g$ is a sum of $g_i$'s. Thus $\conv_\cR(g)$ is a sum of the convex envelopes $\conv_{\cI_i}(g_i)$'s \cite{Falk-MS1969}. Moreover, $\conv_{\cI_i}(g_i)(z_i)=l_iu_i-(l_i+u_i)z_i$, that is, $\conv_{\cI_i}(g_i)$ is affine and agrees with $g_i$ at $l_i$ and $u_i$.
\end{proof}
Note that our interest is in a lower bound of $g$ over the feasible set $\cF^\circ\cap\cR$. Since $\conv_\cR(g)(z)\leq g(z)$ for any $z\in \cR$ and of course for any $z\in \cF^\circ\cap \cR$, the lower bound of $g(z)$ over $\cF^\circ\cap \cR$ can be obtained by solving
\begin{align}
&\min_{z\in \cF^\circ\cap \cR} \conv_\cR(g)(z) \\
\Leftrightarrow & \min_{z\in \cF^\circ\cap \cR} \sum\nolimits_{i=1}^nl_iu_i-(l_i+u_i)z_i \label{eq:lb_z} \\
\Leftrightarrow &\max_{B\in \cR_B}  \sum_{i=1}^n\nolimits(l_i+u_i)\trsp{v_i}\text{vec}(B). \label{eq:lb_constrained}
\end{align}
Going from \eqref{eq:lb_z} to \eqref{eq:lb_constrained} is rewriting the constraint $z\in\cF^\circ\cap \cR$ on $z$ of \eqref{eq:lb_z} into $B\in\cR_B$ on $B$ of \eqref{eq:lb_constrained}, where we define
	\begin{align}
	\cR_B=\{B\in\conv(\cP):\trsp{V}\text{vec}(B)\in\cR \}.
	\end{align}
Solving the linear program \eqref{eq:lb_constrained} is still expensive in practice.
Instead, following \cite{Lian-PAMI17}, we solve \eqref{eq:lb_constrained} over the superset  $\conv(\cP)$ of $\cR_B$ by linear assignment algorithms in $\cO(m^3)$ time \cite{Jonker-Computing1987}, trading tightness for efficiency. 

\subsection{The Upper Bound Computation}\label{subsection:ub}
Typically the optimal solution to \eqref{eq:lb_constrained} is used to compute an upper bound of the optimal value of $f$. We refine this strategy. Having solved \eqref{eq:lb_constrained} over $\conv(\cP)$, we use this solution $\underline{\Pi}_0$ as initialization to solve \eqref{eq:MLE} via alternating minimization to get $\underline{\Pi}$ and $f(\underline{\Pi})$ as the upper bound. Then $f(\underline{\Pi}) \leq f(\underline{\Pi}_0)$. Alternating minimization is of complexity $\cO(m\log(m)+n^2m)$ per iteration \cite{Peng-ICASSP2019} and in our experiments it typically terminates within about $50$ iterations for $n\leq 8, \, m\leq 1000$.




\section{Numerical Results}\label{section:numerics}
We compare our algorithm\footnote{Full code available at \url{https://github.com/liangzu/CCVMIN}.} (\texttt{CCV-Min}) with several existing methods: the branch-and-bound algorithm with dynamic programming \cite{Tsakiris-ICML2019}-A that globally minimizes \eqref{eq:MLE}, the RANSAC-type scheme \cite{Tsakiris-ICML2019}-B, the algebraic-geometric solution based on Gr{\"o}bner basis computation \cite{Tsakiris-arXiv18v2}, the convex $\ell_1$ robust regression \cite{Slawski-JoS19}, and the pseudo-likelihood method \cite{Slawski-arXiv2019b}.

\begin{figure*}[t!]
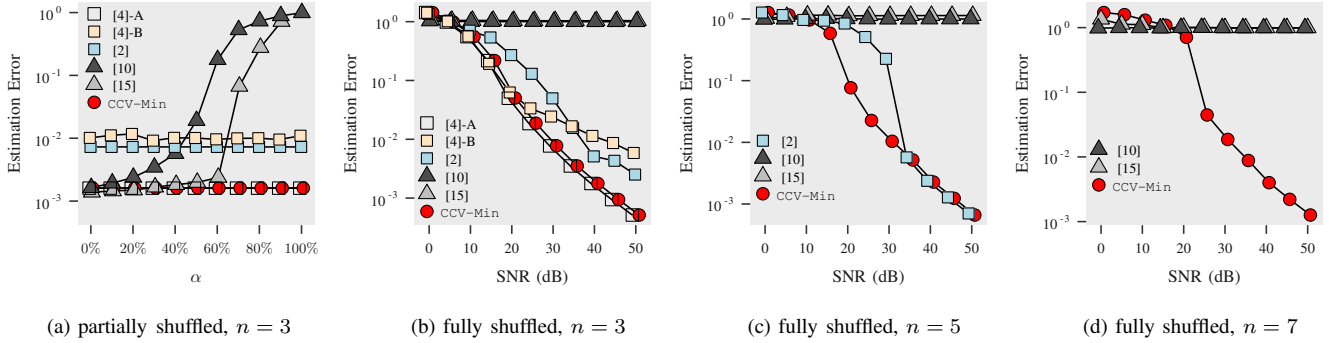

	\centering
	\subfloat[partially shuffled, $n=3$]{\input{./figures/partially_shuffled_randn.tex} \label{fig:partially_shuffled_performance_randn} }
	\subfloat[fully shuffled, $n=3$]{\hspace*{-0.35cm}\input{./figures/fully_shuffled_n3.tex} \label{fig:fully_shuffled_performance_n3} }
	\subfloat[fully shuffled, $n=5$]{\hspace*{-0.35cm}\input{./figures/fully_shuffled_n5.tex} \label{fig:fully_shuffled_performance_n5} }
	\subfloat[fully shuffled, $n=7$]{\hspace*{-0.35cm}\input{./figures/fully_shuffled_n7.tex} \label{fig:fully_shuffled_performance_n7} }
	\caption{Estimation errors of the compared methods for fully and partially shuffled data with $m=100$ fixed. In (a) $\text{SNR}=40\text{dB}$.}
\end{figure*}
\myparagraph{Experiments on synthetic data}
We generate synthetic data as in \cite{Tsakiris-ICML2019}. Entries of the $m\times n$ matrix $A$ and the signal $x^*\in \bbR^n$ are randomly sampled from the standard normal distribution. The vector $y$ is then obtained by 1) randomly shuffling fixed percentage $\alpha$ of entries of $Ax^*$, where $0\leq \alpha\leq 100\%$ and 2) contaminating it with additive noise sampled from the normal distribution $\cN(0,\sigma^2 I_m)$. We evaluate the algorithms by the relative estimation error of $x^*$, $\frac{\norm{\bar{x}-x^*}}{\norm{x^*}}$, with $\bar{x}$ the algorithmic output.\footnote{\texttt{CCV-Min} returns a permutation, so we get $\bar{x}$ via least-squares.} We also report running times of the algorithms.\footnote{Experiments are run on
	an Intel(R) i7-8650U, 1.9GHz, 16GB machine.}

Fig. \ref{fig:partially_shuffled_performance_randn} shows the estimation errors for different percentages $\alpha$ of shuffled data, with $m=100,n=3$. The errors of algorithms \cite{Slawski-JoS19} and \cite{Slawski-arXiv2019b} increase from about $0.1\%$ to $100\%$ as $\alpha$ grows, with breaking points $\alpha= 50\%$ and $\alpha=70\%$ respectively. The other algorithms are immune to the shuffled ratios. Specifically, \cite{Tsakiris-ICML2019}-B and \cite{Tsakiris-arXiv18v2} result in errors of roughly $1\%$, while \cite{Tsakiris-ICML2019}-A and \texttt{CCV-Min} enjoy errors no more than $0.3\%$. Note that \cite{Tsakiris-ICML2019}-A and \texttt{CCV-Min} have the same performance because they solve the equivalent \eqref{eq:MLE} and \eqref{eq:cvx_max} respectively to global optimality. Fig. \ref{fig:fully_shuffled_performance_n3} depicts the errors under different SNR values for $m=100,n=3,\alpha=100\%$. While \cite{Slawski-JoS19} and \cite{Slawski-arXiv2019b} can not cope with fully shuffled data, the rest four methods exhibit decreased errors as the SNR values increase. Figs. \ref{fig:fully_shuffled_performance_n5} and \ref{fig:fully_shuffled_performance_n7} are produced under the same settings as in Fig. \ref{fig:fully_shuffled_performance_n3} except respectively for $n=5$ and $n=7$. Interestingly, the error curves of \texttt{CCV-Min} follow the same trend $n$ even though $n$ is now larger. Note \cite{Tsakiris-ICML2019}-A, \cite{Tsakiris-ICML2019}-B and \cite{Tsakiris-arXiv18v2} were not included since they need more than $12$ hours to terminate for $n=5$ or $n=7$. On the other hand, \texttt{CCV-Min} terminates in about $7$sec and $6$min respectively. Table \ref{table:fully_shuffled_rt} presents a more comprehensive view of the running times as $n$ varies \footnote{\cite{Slawski-JoS19},\cite{Slawski-arXiv2019b} are not included as they only work for partial mismatches.}. We see that \texttt{CCV-Min} is fast in general, the fastest when $n\geq 5$, and the \textit{only} tractable method for $n=7,8$ in particular. On the other hand its breaking point is $n=9$, which we attribute to the inherent complexity of the branch-and-bound scheme. It is important though to contrast this to the breaking point $n=5$ of \cite{Tsakiris-ICML2019}-A which is also a branch-and-bound method, the advantage of \texttt{CCV-Min} being due to its tight lower bound computation, a consequence of the concave minimization formulation.

\begin{table}
	\centering
	\caption{Running times for different $n$ with $\text{SNR}=40\text{DB},m=100$ fixed. Fully shuffled data.} \label{table:fully_shuffled_rt}
	\begin{tabular}{@{}rrccccc@{}}\toprule
		& \multicolumn{4}{c}{Running Time} \\
		\cmidrule{2-5} 
		$n$ & \cite{Tsakiris-ICML2019}-A & \cite{Tsakiris-ICML2019}-B & \cite{Tsakiris-arXiv18v2} &  \texttt{CCV-Min}    \\ \midrule
		$3$ &            $0.48$sec        &             $37$sec        &        $3$msec              &    $0.42$sec    \\
		$4$ &            $5$sec          &             $17$min        &        $7$msec              &    $2.43$sec    \\
		$5$ &            $>12$hr         &             $>12$hr        &        $43$sec            &    $7.16$sec   \\
		$6$ &                            &                           &        $37$min           &    $72.5$sec    \\
		$7$ &                           &                          &         $>12$hr          &    $6$min   \\
		$8$ &                           &                          &                        &  $40$min  \\
		$9$ &                           &                          &                        &  $>12$hr  \\
		\bottomrule
	\end{tabular}
\end{table}

\begin{figure}
	\centering
	\subfloat{
\begin{tikzpicture}[x=1pt,y=1pt]
\definecolor{fillColor}{RGB}{255,255,255}
\path[use as bounding box,fill=fillColor,fill opacity=0.00] (0,0) rectangle (216.81,101.18);
\begin{scope}
\path[clip] (  0.00,  0.00) rectangle (216.81,101.18);
\definecolor{drawColor}{RGB}{255,255,255}
\definecolor{fillColor}{RGB}{255,255,255}

\path[draw=drawColor,line width= 0.6pt,line join=round,line cap=round,fill=fillColor] (  0.00,  0.00) rectangle (216.81,101.18);
\end{scope}
\begin{scope}
\path[clip] ( 27.02, 25.42) rectangle (211.31, 95.68);
\definecolor{fillColor}{gray}{0.92}

\path[fill=fillColor] ( 27.02, 25.42) rectangle (211.31, 95.68);
\definecolor{drawColor}{RGB}{0,0,0}
\definecolor{fillColor}{RGB}{255,0,0}

\path[draw=drawColor,line width= 0.6pt,line join=round,fill=fillColor] ( 71.18, 28.61) rectangle ( 75.97, 40.29);
\definecolor{fillColor}{RGB}{190,190,190}

\path[draw=drawColor,line width= 0.6pt,line join=round,fill=fillColor] ( 66.38, 28.61) rectangle ( 71.18, 92.48);
\definecolor{fillColor}{gray}{0.30}

\path[draw=drawColor,line width= 0.6pt,line join=round,fill=fillColor] ( 61.58, 28.61) rectangle ( 66.38, 91.00);
\definecolor{fillColor}{RGB}{173,216,230}

\path[draw=drawColor,line width= 0.6pt,line join=round,fill=fillColor] ( 56.78, 28.61) rectangle ( 61.58, 53.26);
\definecolor{fillColor}{RGB}{255,228,196}

\path[draw=drawColor,line width= 0.6pt,line join=round,fill=fillColor] ( 51.98, 28.61) rectangle ( 56.78, 42.45);
\definecolor{fillColor}{RGB}{255,255,255}

\path[draw=drawColor,line width= 0.6pt,line join=round,fill=fillColor] ( 47.18, 28.61) rectangle ( 51.98, 40.38);
\definecolor{fillColor}{RGB}{255,0,0}

\path[draw=drawColor,line width= 0.6pt,line join=round,fill=fillColor] (126.37, 28.61) rectangle (133.56, 32.30);
\definecolor{fillColor}{RGB}{190,190,190}

\path[draw=drawColor,line width= 0.6pt,line join=round,fill=fillColor] (119.17, 28.61) rectangle (126.37, 62.12);
\definecolor{fillColor}{gray}{0.30}

\path[draw=drawColor,line width= 0.6pt,line join=round,fill=fillColor] (111.97, 28.61) rectangle (119.17, 58.62);
\definecolor{fillColor}{RGB}{173,216,230}

\path[draw=drawColor,line width= 0.6pt,line join=round,fill=fillColor] (104.77, 28.61) rectangle (111.97, 36.83);
\definecolor{fillColor}{RGB}{255,0,0}

\path[draw=drawColor,line width= 0.6pt,line join=round,fill=fillColor] (183.95, 28.61) rectangle (191.15, 31.47);
\definecolor{fillColor}{RGB}{190,190,190}

\path[draw=drawColor,line width= 0.6pt,line join=round,fill=fillColor] (176.76, 28.61) rectangle (183.95, 73.74);
\definecolor{fillColor}{gray}{0.30}

\path[draw=drawColor,line width= 0.6pt,line join=round,fill=fillColor] (169.56, 28.61) rectangle (176.76, 64.36);
\definecolor{fillColor}{RGB}{173,216,230}

\path[draw=drawColor,line width= 0.6pt,line join=round,fill=fillColor] (162.36, 28.61) rectangle (169.56, 48.61);
\end{scope}
\begin{scope}
\path[clip] (  0.00,  0.00) rectangle (216.81,101.18);
\definecolor{drawColor}{gray}{0.10}

\node[text=drawColor,anchor=base east,inner sep=0pt, outer sep=0pt, scale=  0.55] at ( 22.07, 26.72) {$0\%$};

\node[text=drawColor,anchor=base east,inner sep=0pt, outer sep=0pt, scale=  0.55] at ( 22.07, 47.09) {$1\%$};

\node[text=drawColor,anchor=base east,inner sep=0pt, outer sep=0pt, scale=  0.55] at ( 22.07, 67.45) {$2\%$};

\node[text=drawColor,anchor=base east,inner sep=0pt, outer sep=0pt, scale=  0.55] at ( 22.07, 87.82) {$3\%$};
\end{scope}
\begin{scope}
\path[clip] (  0.00,  0.00) rectangle (216.81,101.18);
\definecolor{drawColor}{gray}{0.20}

\path[draw=drawColor,line width= 0.6pt,line join=round] ( 24.27, 28.61) --
	( 27.02, 28.61);

\path[draw=drawColor,line width= 0.6pt,line join=round] ( 24.27, 48.98) --
	( 27.02, 48.98);

\path[draw=drawColor,line width= 0.6pt,line join=round] ( 24.27, 69.34) --
	( 27.02, 69.34);

\path[draw=drawColor,line width= 0.6pt,line join=round] ( 24.27, 89.71) --
	( 27.02, 89.71);
\end{scope}
\begin{scope}
\path[clip] (  0.00,  0.00) rectangle (216.81,101.18);
\definecolor{drawColor}{gray}{0.20}

\path[draw=drawColor,line width= 0.6pt,line join=round] ( 61.58, 22.67) --
	( 61.58, 25.42);

\path[draw=drawColor,line width= 0.6pt,line join=round] (119.17, 22.67) --
	(119.17, 25.42);

\path[draw=drawColor,line width= 0.6pt,line join=round] (176.76, 22.67) --
	(176.76, 25.42);
\end{scope}
\begin{scope}
\path[clip] (  0.00,  0.00) rectangle (216.81,101.18);
\definecolor{drawColor}{gray}{0.10}

\node[text=drawColor,anchor=base,inner sep=0pt, outer sep=0pt, scale=  0.55] at ( 61.58, 16.68) {$n\leq 4$};

\node[text=drawColor,anchor=base,inner sep=0pt, outer sep=0pt, scale=  0.55] at (119.17, 16.68) {$n=5$};

\node[text=drawColor,anchor=base,inner sep=0pt, outer sep=0pt, scale=  0.55] at (176.76, 16.68) {$n=6$};
\end{scope}
\begin{scope}
\path[clip] (  0.00,  0.00) rectangle (216.81,101.18);
\definecolor{drawColor}{gray}{0.10}

\node[text=drawColor,rotate= 90.00,anchor=base,inner sep=0pt, outer sep=0pt, scale=  0.66] at ( 10.05, 60.55) {Residual Error};
\end{scope}
\begin{scope}
\path[clip] (  0.00,  0.00) rectangle (216.81,101.18);
\definecolor{drawColor}{RGB}{0,0,0}
\definecolor{fillColor}{RGB}{255,255,255}

\path[draw=drawColor,line width= 0.6pt,line cap=round,fill=fillColor] (129.48, 87.60) rectangle (134.20, 92.32);
\end{scope}
\begin{scope}
\path[clip] (  0.00,  0.00) rectangle (216.81,101.18);
\definecolor{drawColor}{RGB}{0,0,0}
\definecolor{fillColor}{RGB}{255,228,196}

\path[draw=drawColor,line width= 0.6pt,line cap=round,fill=fillColor] (129.48, 81.46) rectangle (134.20, 86.18);
\end{scope}
\begin{scope}
\path[clip] (  0.00,  0.00) rectangle (216.81,101.18);
\definecolor{drawColor}{RGB}{0,0,0}
\definecolor{fillColor}{RGB}{173,216,230}

\path[draw=drawColor,line width= 0.6pt,line cap=round,fill=fillColor] (129.48, 75.31) rectangle (134.20, 80.03);
\end{scope}
\begin{scope}
\path[clip] (  0.00,  0.00) rectangle (216.81,101.18);
\definecolor{drawColor}{RGB}{0,0,0}
\definecolor{fillColor}{gray}{0.30}

\path[draw=drawColor,line width= 0.6pt,line cap=round,fill=fillColor] (129.48, 69.17) rectangle (134.20, 73.89);
\end{scope}
\begin{scope}
\path[clip] (  0.00,  0.00) rectangle (216.81,101.18);
\definecolor{drawColor}{RGB}{0,0,0}
\definecolor{fillColor}{RGB}{190,190,190}

\path[draw=drawColor,line width= 0.6pt,line cap=round,fill=fillColor] (129.48, 63.03) rectangle (134.20, 67.75);
\end{scope}
\begin{scope}
\path[clip] (  0.00,  0.00) rectangle (216.81,101.18);
\definecolor{drawColor}{RGB}{0,0,0}
\definecolor{fillColor}{RGB}{255,0,0}

\path[draw=drawColor,line width= 0.6pt,line cap=round,fill=fillColor] (129.48, 56.88) rectangle (134.20, 61.60);
\end{scope}
\begin{scope}
\path[clip] (  0.00,  0.00) rectangle (216.81,101.18);
\definecolor{drawColor}{RGB}{0,0,0}

\node[text=drawColor,anchor=base west,inner sep=0pt, outer sep=0pt, scale=  0.53] at (137.08, 88.14) {\cite{Tsakiris-ICML2019}-A};
\end{scope}
\begin{scope}
\path[clip] (  0.00,  0.00) rectangle (216.81,101.18);
\definecolor{drawColor}{RGB}{0,0,0}

\node[text=drawColor,anchor=base west,inner sep=0pt, outer sep=0pt, scale=  0.53] at (137.08, 82.00) {\cite{Tsakiris-ICML2019}-B};
\end{scope}
\begin{scope}
\path[clip] (  0.00,  0.00) rectangle (216.81,101.18);
\definecolor{drawColor}{RGB}{0,0,0}

\node[text=drawColor,anchor=base west,inner sep=0pt, outer sep=0pt, scale=  0.53] at (137.08, 75.85) {\cite{Tsakiris-arXiv18v2}};
\end{scope}
\begin{scope}
\path[clip] (  0.00,  0.00) rectangle (216.81,101.18);
\definecolor{drawColor}{RGB}{0,0,0}

\node[text=drawColor,anchor=base west,inner sep=0pt, outer sep=0pt, scale=  0.53] at (137.08, 69.71) {\cite{Slawski-JoS19}};
\end{scope}
\begin{scope}
\path[clip] (  0.00,  0.00) rectangle (216.81,101.18);
\definecolor{drawColor}{RGB}{0,0,0}

\node[text=drawColor,anchor=base west,inner sep=0pt, outer sep=0pt, scale=  0.53] at (137.08, 63.57) {\cite{Slawski-arXiv2019b}};
\end{scope}
\begin{scope}
\path[clip] (  0.00,  0.00) rectangle (216.81,101.18);
\definecolor{drawColor}{RGB}{0,0,0}

\node[text=drawColor,anchor=base west,inner sep=0pt, outer sep=0pt, scale=  0.53] at (137.08, 57.43) {\texttt{CCV-Min}};
\end{scope}
\end{tikzpicture} \label{fig:college_MLR} }
	\vspace*{-0.5cm}
	\caption{Residual errors on the real data \cite{brase-book2011} for different $n$'s.}
\end{figure}
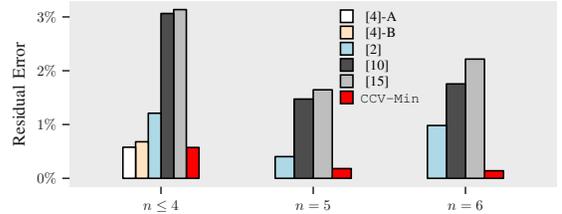
\myparagraph{Experiments on real data}
We use eleven linear regression datasets $\{y^{(i)}, A^{(i)}\}_{i=1}^{11}$ from \cite{brase-book2011}, arising in contexts such as
predicting the blood pressure from the age and weight, box office receipts of Hollywood movies from production and promotional costs, or the final scores for General Psychology from previous exams. The regression orders $n^{(i)}$ take values in $\{2,\dots, 6\}$ and the number $m^{(i)}$ of samples varies from $10$ to $60$. To generate mismatches we randomly fully shuffle the rows of $A^{(i)}$. Since the ground truth $x^*$ is not available, we use the averaged minimal residual error $\frac{1}{m^{(i)}\norm{y^{(i)}}}\min_{\Pi\in \cP}\big|\big|\Pi  y^{(i)}-A^{(i)} \bar{x}\big| \big|_2$. This is plotted in Fig. \ref{fig:college_MLR}. An immediate observation
is that all methods, even the ones that are expected to succeed only with partial mismatches \cite{Slawski-JoS19}, \cite{Slawski-arXiv2019b}, all perform quite well with errors roughly not more than $3\%$. This is because features across different samples appear to be highly correlated, so that the effect of the permutation is only mild\footnote{A similar phenomenon has been observed in \cite{Slawski-arXiv2019b}.}. Be as it may \texttt{CCV-Min} consistently gives the smallest errors\footnote{Of the same order as those of standard linear regression.}.

\newpage
\label{sec:refs}
%

\end{document}